\documentclass[12pt,a4paper]{article}
\usepackage[latin1]{inputenc}
\usepackage{amsmath}
\usepackage{amsthm}
\usepackage{amsfonts}
\usepackage{amssymb}
\usepackage{hyperref}
\usepackage{cleveref}
\usepackage{graphicx}
\usepackage{nicefrac}
\usepackage{fullpage}
\usepackage{bm}
\usepackage{xcolor}
\usepackage{tikz}

\newtheorem{thm}{Theorem}[section]
\newtheorem{prop}[thm]{Proposition}
\newtheorem{lemma}[thm]{Lemma}

\newtheorem{cor}[thm]{Corollary}

\newtheorem{claim}[thm]{Claim}
\newtheorem{defi}[thm]{Definition}
\newtheorem{cnst}[thm]{Construction}
\newtheorem{obs}[thm]{Observation}

\newcommand{\ceil}[1]{\left \lceil #1 \right \rceil}
\newcommand{\floor}[1]{\left \lfloor #1 \right \rfloor}

\newcommand{\Fp}{\mathbb{F}_p}
\newcommand{\cC}{\mathcal{C}}
\newcommand{\Fpm}{\mathbb{F}_{p^m}}

\newcommand{\lcm}{\textup{lcm}}
\newcommand{\Leak}{\textup{Leak}_{\Theta, \tau}}

\newcommand\blfootnote[1]{%
	\begingroup
	\renewcommand\thefootnote{}\footnote{#1}%
	\addtocounter{footnote}{-1}%
	\endgroup
}

\begin{document}
	\title{Nonlinear Repair  of Reed-Solomon Codes.}
	
	\author{
		Roni Con\thanks{Department of Computer Science, Tel Aviv University, Tel Aviv, Israel. Email: roni.con93@gmail.com} \and Itzhak Tamo \thanks{Department of Electrical Engineering-Systems, Tel Aviv University, Tel Aviv, Israel. Email: tamo@tauex.tau.ac.il}
	} 
	\date{}
	\maketitle

	\begin{abstract}
		The problem of repairing linear codes and, in particular, Reed Solomon (RS) codes has attracted a lot of attention in recent years due to their extreme importance to distributed storage systems. In this problem, a failed code symbol (node) needs to be repaired   by downloading as little information as possible from
		a subset of the remaining nodes.  By now, there are examples of RS codes that have efficient repair schemes, and some even attain the cut-set bound. However, these schemes fall short in several aspects; they require a considerable field extension degree. They do not provide any nontrivial repair scheme over prime fields. Lastly, they are all linear repairs, i.e., the computed functions are linear over the base field.  Motivated by these and by a question raised in \cite{guruswami2017repairing} on the power of nonlinear repair schemes, we study the problem of nonlinear repair schemes of RS codes.
		
		Our main results are
		the first nonlinear repair scheme of RS codes with asymptotically optimal repair bandwidth (asymptotically matching the cut-set bound). This is the first example of a nonlinear repair scheme of any code and also the first example that a nonlinear repair scheme can outperform all linear ones.
		Lastly, we show that the cut-set bound for RS codes is not tight over prime fields by proving a tighter bound, using additive combinatorics ideas. 
		\blfootnote{This work was partially supported by the European Research Council (ERC grant number 852953) and by the Israel Science Foundation (ISF grant number 1030/15).}
		
		
		
	\end{abstract}
	
	\newpage
	\tableofcontents
	\newpage
	
	\section{Introduction}
	A distributed storage system is a form of a computer network that stores information reliably across multiple storage devices while introducing redundancy (for increased reliability) in the form of an erasure-correcting code. Typically, a file that needs to be stored in the system is cut into $k$ data fragments that form the input to an erasure correcting code. Then,   $n-k$ new redundant fragments are computed, and the total $n$ fragments are stored across distinct storage devices (nodes). 
	
	These days, large-scale distributed storage systems face many new challenges that rose recently in the information era, characterized by    
	enormous data being generated daily. One such challenge is the problem
	of efficiently (measured in terms of system resources) repairing a single failed node \cite[Section 6.6]{rashmi2014hitchhiker}.  
	In such a scenario, the lost data stored on the failed node  needs to be recovered. Since an erasure-correcting code is employed, this amounts to  decoding a single codeword erasure. This problem called the \emph{(exact) repair problem} was first considered in the seminal paper of Dimakis et al. \cite{dimakis2010network} and has since witnessed an explosive amount of research.  
	
	A  \emph{repair scheme} for the repair problem is a method  to recover a failed node's data from the data stored on the remaining non failed ones. 
	Assuming the storage system employs an $[n,k]$ maximum distance separable (MDS) code, which allows the original file to be recovered from any $k$  nodes, then  a \emph{trivial scheme} would be to contact any $k$ of the remaining $n-1$ nodes, recover the original file and thus also recover the lost data. 
	However, this scheme is too costly in terms of the amount of information transmitted across the network, termed as the \emph{repair bandwidth}. This parameter is the main figure of merit that one wishes to optimize when constructing efficient repair schemes. Many code constructions and ingenious repair schemes were constructed over the years since the introduction of the repair problem 
	\cite{dimakis2010network, el2010fractional,goparaju2013data,papailiopoulos2013repair,tamo2012zigzag,wang2016explicit,rashmi2011optimal,ye2017explicithighrate,ye2017explicit,goparaju2017minimum}.
	However, despite considerable progress, there are several open questions and challenges left to overcome. In this paper, we address and resolve some of them.
	
	

	\subsection{The setup and previous results}
	Before formally defining and discussing RS codes' repair problem, we begin with basic definitions of linear  codes.  
	
	An $[n,k]$ \emph{array code} $\cC$, with \emph{subpacketization} $L$ over a finite field $\mathbb{F}$ is a linear subspace $\cC\subseteq \mathbb{F}^{L \times n}$ over $\mathbb{F}$ and dimension $k L$. The \emph{length} and the \emph{rate} of the code are  $n$ and $k/n$, respectively, and the elements of $\cC$ are called \emph{codewords}. The  $i$th  (code) symbol of a codeword   is the $i$th  column of the codeword, which is  a vector of length $L$ over $\mathbb{F}$. An $[n,k]$ array code is called MDS if each codeword is uniquely determined by any $k$ of its symbols. Lastly, \emph{scalar codes} (which are the more common mathematical object when one refers to a code) are   array codes  with $L=1$. Hence, any scalar code is also an array code. Furthermore, a scalar code $\cC$ over a field extension $\mathbb{E}$ with a subfield $\mathbb{F}$, where $[\mathbb{E}:\mathbb{F}]=L$ can be viewed as an array code over $\mathbb{F}$ and subpacketization $L$, simply by expanding each symbol of the code to a column vector of length $L$ over $\mathbb{F}$, according to some basis of $\mathbb{E}$ over $\mathbb{F}$.    
	
	In a distributed storage system that employs an array code of length $n$, it is assumed that each code symbol resides on a distinct storage device (node) to increase the data reliability in case of a node malfunctioning. As already mentioned, the most common scenario in such systems is a single node's failure, which is a single symbol erasure in the coding theory terminology. Such an event triggers a repair scheme whose goal is to recover the erased symbol by receiving information from a subset of the remaining $n-1$ nodes, called \emph{helper nodes}. We mention that we interchangeably use the term node and a symbol of the code in the sequel and say that node $i$ holds (stores) the $i$th symbol.  The figure of merit considered in this problem is the total incurred bandwidth across the network due to the repair scheme. In other words, how many bits the helper nodes need to transmit to repair the failed node.   
	This quantity is called the \emph{repair bandwidth} whose formal definition  is given next. 
	\begin{defi} [Repair bandwidth] \label{def:repair-bandwidth}
		Let $\cC$ be an $[n,k]$ MDS array code. 
		For $i\in [n]$ and a subset $D\subseteq [n]\setminus \{i\}, \left| D \right|=d \geq k$ of helper nodes, define $N(i, D)$ as the smallest number of bits that need to be transmitted  from the helper nodes $D$ in order to repair the failed node $i$. The repair bandwidth of the code $\cC$ with $d$ helper nodes is defined as  
		\[
		\max_{\substack{i\in [n]\\ D\subseteq [n]\setminus\{i\}, \left|D\right| = d}} N(i,D) \;.
		\]
	\end{defi}
	Note that the transmitted information from each helper node can be any function of the symbol it holds.  Next, we shall state the well-known lower bound on the repair bandwidth, called the \emph{cut-set bound}, derived by Dimakis et al. in the seminal work \cite{dimakis2010network}.
	\begin{thm} \cite{dimakis2010network} \label{thm:cut-set-bound}
		Let $\cC$ be an $[n,k]$ MDS array code with subpacketization $L$ over a field $\mathbb{F}$, then for any $i\in [n]$ and any set of $d$ helper nodes  $D\subseteq [n]\setminus \{i\}$ 
		\begin{equation} \label{eq:cut-set}
		N(i, D) \geq \frac{d L \log(|\mathbb{F}|)}{d+1 -k}\;.
		\end{equation}
	\end{thm}
	MDS codes are widely used in practice due to their optimal resiliency to erasures for the given amount of added redundancy. Therefore, MDS codes that also attain the cut-set bound during the repair scheme are highly desirable.  
	An $[n,k]$ MDS code achieving the cut-set bound  \eqref{eq:cut-set} with equality for any failed node $i$ by any $d$ helper nodes is called an \emph{$[n, k, d]$ MSR (minimum storage regenerating)} (array) code. 
	
	RS codes are the most known MDS codes, and they have found many applications both in  theory and practice (some   applications include QR codes \cite{soon2008qr}, Secret sharing schemes \cite{mceliece1981sharing}, space transmission
	\cite{wicker1999reed}, encoding data on CDs  \cite{wicker1999reed} and more). 
	The ubiquity  of these codes can be attributed to their  simplicity and their efficient encoding and decoding algorithms.  	 Thus, it might be surprising that RS codes were  not considered a possible solution for the repair problem in distributed storage systems. 
	In fact, many researchers  believed that RS codes do not admit efficient repair schemes except for the trivial scheme.
	Therefore, several MSR (that are not RS codes)
	codes were constructed, e.g., \cite{ye2017explicithighrate,ye2017explicit,rashmi2011optimal,wang2016explicit,goparaju2017minimum,raviv2017constructions}. 
	Yet, the problem of understanding the efficiency of the repair of RS codes remained unresolved, and since many distributed storage systems in fact employ RS codes (e.g., Facebook Hadoop Analytics cluster employs a  $[14, 10]$ RS code
	\cite{sathiamoorthy2013xoring}), this problem assumed even greater importance.

	
	In \cite{guruswami2017repairing} it was shown that the repair problem of RS codes could be seen as a new and interesting twist on the standard interpolation problem of polynomials. Thus, studying this problem  might have theoretical implications since polynomial interpolation is widely used across all areas of mathematics. 
	Before explaining the new twist on the interpolation problem, we give next a formal definition of RS codes.
	
	\begin{defi}
		Let $\alpha_1, \alpha_2, \ldots, \alpha_n$ be distinct points of the finite field $\mathbb{F}_q$ of order  $q$. For $k<n$ the $[n,k]_q$ \emph{RS code} 
		defined  by the  evaluation set $\{ \alpha_1, \ldots, \alpha_n \}$ is the   set of codewords 
		\[
		\left \lbrace \left( f(\alpha_1), \ldots, f(\alpha_n) \right) \mid f\in \mathbb{F}_q[x],\deg f < k \right \rbrace \;.
		\]
		When $n = q$, the resulting  code is called a \emph{full-length RS code}.
		
	\end{defi}
	
	Thus, a codeword of an $[n, k]$ RS code is the evaluation vector of some polynomial of degree
	less than $k$ at $n$ distinct points, i.e., the codeword that corresponds  to a polynomial $f$ of
	degree less than $k$ is $(f(\alpha_1),\ldots, f(\alpha_{n}))$. Assuming that the $n$th node has failed, i.e., the value  of $f(\alpha_n)$ is lost, the question boils down to how much  information is needed from the remaining nodes storing
	$f(\alpha_1),\ldots, f(\alpha_{n-1})$, in order to determine $f(\alpha_n)$. It is clear that any $k$ values of $f(\alpha_i)$ suffice to recover the polynomial $f$, and in particular recover  $f(\alpha_n)$. In the terminology of a repair scheme, this corresponds to $d=k$ helper nodes that transmit their entire symbol. This type of repair scheme is termed as the \emph{trivial repair}, although it also attains  the cut-set bound \eqref{eq:cut-set} with equality. 
	
	Hence, the more exciting and challenging question is whether it is possible to recover the polynomial value at a specific location without recovering the original polynomial, thereby possibly requiring less information from the $d$ helper nodes for $d>k$.
	It turns out that to determine $f(\alpha_n)$, one needs much less information than the amount needed in the trivial scheme that employs  polynomial interpolation.
	Indeed, Shanmugam, Papailiopoulos, Dimakis, and Caire \cite{shanmugam2014repair} 
	developed a general framework for
	repairing  scalar MDS codes and, in particular, RS codes, then they exemplified their framework by showing that there are repair schemes for RS codes that are more efficient than the trivial scheme.  
	Then, Guruswami and Wootters \cite{guruswami2017repairing} generalized the framework of \cite{shanmugam2014repair} and gave a complete characterization of linear  repair schemes of scalar MDS codes. They also provided few examples of RS codes with linear repair schemes that outperform the trivial repair scheme. 
	A more recent work by Tamo, Ye, and Barg \cite{tamo2017optimal} used the framework of \cite{guruswami2017repairing} and showed that for every $k < d < n$ there are RS codes that are indeed $[n,k,d]$ MSR codes. The caveat in their work is the large field extension degree (subpacketization) which is $L = \exp((1 + o(1))n \log(n))$. 
	
	Unfortunately, in \cite{tamo2017optimal}, the authors provided an almost matching lower bound of $L = \exp(\Omega(k \log(k)))$ on the degree of the field extension that is required in order for an RS code (and in general any scalar MDS code) to be an MSR code with a linear repair scheme. The results of \cite{guruswami2017repairing} were extended even further in \cite{dau2018repairing} and \cite{mardia2018repairing},  to consider multiple node failures. 
	As a final remark, the strong lower bound on the field size for linear repair schemes, given in \cite{tamo2017optimal}, gives a clear motivation for studying nonlinear repair schemes.

	\subsection{Linear and nonlinear repair schemes}
	As already discussed, any code, either scalar or array, can be viewed as an array code over some prime field $\Fp$. Therefore, consider an $[n,k]$ array code $\cC$ over $\Fp$ and  subpacketization $L$. We say that a repair scheme is linear if each computed function $\mu$ by a helper node is linear over the prime field $\Fp$. Equivalently, if for any $\alpha,\beta\in \Fp$ and $x,y\in \Fp^{L}$ $$\mu(\alpha x+\beta y)=\alpha\mu(x)+\beta \mu(y).$$
	
	All the existing efficient repair schemes of linear codes rely on the fact that  if they are viewed as array codes, their subpacketization level $L$ is extremely large. On the other hand, it is known that this is indeed needed if one employs linear repair schemes \cite{alrabiah2019exponential,tamo2017optimal}. Hence we are motivated to study nonlinear repair schemes. Furthermore, Guruswami and Wootters asked the following  in  \cite{guruswami2017repairing}: ``How much better can one do with nonlinear repair schemes?"
	
	
	A good starting point for constructing nonlinear repair schemes is to consider array codes with subpacketization $L=1$. For RS codes, this means to be defined over a prime field. 
	In such a case, any nontrivial repair scheme must   be nonlinear. Indeed, any nonzero linear function $\mu:\Fp\rightarrow \Fp$  must be bijective, which in terms of the transmitted information means that the helper node sends its entire symbol, and note that this corresponds to the trivial repair scheme.
	We conclude that any linear repair scheme is  equivalent to the trivial repair scheme for RS codes over a prime field. Thus, any improvement over the trivial repair scheme must be nonlinear, and it automatically implies that it outperforms all linear repair schemes.  
	
	Lastly, repairing RS codes is strongly related to leakage resilience of Shamir's secret sharing scheme. In the recent work of Benhamouda, Degwekar,   Ishai,  and  Rabin \cite{benhamouda2018local}, the authors showed that in certain parameters regimes, Shamir's secret sharing scheme over prime fields is leakage resilient (see \Cref{SSS} for the exact definitions). Thus, studying  nonlinear repair schemes of RS codes over prime fields could have implications outside the scope of coding for distributed storage.

	\subsection{Our contribution}
	In this paper, we make the first step towards understanding nonlinear repair schemes' power for linear codes. In particular,  we present a nonlinear repair scheme of  RS codes that outperforms all the linear ones. For all we know,  this is the first nonlinear repair scheme.
	
	In the real-world scenario of this problem,   large files are stored across a relatively small number of nodes. Therefore, each node stores  a large chunk of the file.  Hence, in this work, we think of $k,d$, and $n$ as small constants, while the alphabet size $p$ (which corresponds to the data stored by each node) is large and tends to infinity. We note that this point of view is different than the usual point of view in coding theory. In light of that, we say that a code symbol admits an asymptotically optimal repair (bandwidth) if the repair bandwidth tends to the cut-set bound \eqref{eq:cut-set}  as $p$ tends to infinity, and $n$ is fixed. 
	
	Below, we summarize the main contributions of this paper, and where $d$ is the number of helper nodes.

	\begin{enumerate}
		\item For any $2< d< n$, we show that $1-o(1)$  fraction of all the $[n,2]_p$ RS codes are asymptotically MSR codes, where the term $o(1)$ tends to zero as  $p$ tends to infinity. Namely, any code symbol admits an asymptotically optimal repair by any set of $d$ helper nodes. As a byproduct, this implies that nonlinear repair schemes can outperform all the linear ones. 
		\item We show that the phenomenon of RS codes with nonlinear repair schemes that outperform all the linear ones also holds over infinitely many field extensions. 
		\item For any $k<d\leq n/2$ we present an explicit construction of an $[n,k]_p$ RS code such that its symbols can be partitioned into two sets of equal size, such  that each symbol  admits an  asymptotically optimal repair using any $d$ helper nodes from the other set. 
		\item We show that any  full-length RS code over the prime  field $\Fp$ exhibits some efficient repair properties. 
		Specifically, for any $k<d$ and large enough $p$, we show that each node has  $\Omega(p)$ distinct sets of helper nodes of size $d$ that can repair it with asymptotically optimal repair bandwidth.

		
		
		
		\item Unlike the problem of repairing RS codes over field extensions, we show that one can not achieve the cut-set bound with equality, and over prime fields, one can obtain a tighter lower bound on the total incurred bandwidth. Concretely, we improve the cut-set bound (in the symmetric case, details below) by   showing  that every node must transmit another additive factor of $\Omega(\log(k)/(d-k+1))$ bits.
	\end{enumerate}

	\subsection{Organization of the paper}
	In \Cref{sec:general-rapair}, we present a general framework for repairing a failed node and obtain a necessary condition for a successful repair. In \Cref{sec:exsitence-results}, we show the existence of $[n,2]_p$ RS codes, which are asymptotically MSR codes, then extend this result to field extensions. In \Cref{stam-section}, we turn to explicit constructions of RS codes with efficient repair schemes.  
	We complement the achievability results (code constructions) given in the previous section by improving the cut-set bound in \Cref{sec:new-cut-set-bound}.
	In \Cref{SSS}, we discuss the implications of  our results on repairing RS codes for  leakage-resilient of Shamir's secret sharing scheme over prime fields. 
	We conclude in \Cref{sec:concluding remarks} with open questions.
	\section{A General framework for node repair} \label{sec:general-rapair}
	Throughout, let $n,k$, and $d$ be the number of nodes (code's length), the RS code dimension, and the number of helper nodes, respectively. Also, let $p$ be a prime number, where we think of $n,k$, and $d$ as constants, and $p$ tends to infinity. 
	
	This section describes a general  repair framework for repairing a single failed node that applies to any repair scheme. For simplicity, we will assume that the last node, i.e., the $n$th node is the failed node that needs to be repaired using all the remaining $n-1$ other nodes, i.e., $d=n-1$. To simplify the notation even further, we assume symmetry between the nodes in terms of the amount of information transmitted, i.e., each helper node transmits the same amount of information. We note that the framework can be easily generalized to the most general case,  i.e., an arbitrary failed node, an arbitrary number of helper nodes $k\leq d\leq n-1$, and the non-symmetric case.  Lastly, we would like to emphasize that the model assumes no errors; namely, the received information from the helper nodes is error-free. One can generalize this model by removing this assumption, as it was done in \cite{rashmi2012regenerating,silberstein2015error}.
	
	We begin with some needed notations. For positive integers $a<b$ let $[a, b]=\{a,a+1, \ldots, b\}$ and $[a]=\{1,2,\ldots, a\}$. Throughout, let  $p$ be a prime number and for a positive integer $m$ let  $\Fpm$  be the  finite field of size $p^m$.
	An arithmetic progression in some field $\mathbb{F}$ of length $N$ and a step $s\in \mathbb{F}$ is a set of the form $\{a, a+s, \ldots, a+(N-1)s \}$ for some $a \in \mathbb{F}$.
	For two sets $A,B\subseteq \Fp$,  define their \emph{sumset} as  
	$ A + B := \{ a+b \mid a\in A, b\in B \}$. 
	For an element $\gamma \in \Fp$ we denote by $\gamma \cdot A:=\{\gamma \cdot a : a\in A \}$ all the possible products of $\gamma$ with elements in $A.$
	
	
	
	We are now ready to present the general repair  framework. Let $\cC \subseteq \Fp^n$ be a linear code over $\Fp$. A repair scheme for its $n$th symbol is a set of $n-1$ functions $\mu_i: \Fp \rightarrow [s]$ and a function $G:[s]^{n-1}\rightarrow \Fp$ such that for any codeword $(c_1,\ldots,c_n)\in \cC$ 
	\begin{equation}
	\label{stamstam}
	G(\mu_{1}(c_1),\ldots,\mu_{n-1}(c_{n-1}))=c_n.
	\end{equation}
	Upon a failure of the $n$th symbol, the $i$th symbol, which holds the symbol $c_i$ computes $\mu_i(c_i)$ and transmits it over the network using $\lceil \log s\rceil $ bits. Upon receiving  the $n-1$ messages $\mu_i(c_i)$, the repair scheme is completed by calculating the $n$th symbol using   \eqref{stamstam}. 
	The  bandwidth of the repair scheme, which is the total number of bits transmitted across the network during the repair, equals $(n-1)\ceil{\log(s)}$ since every node transmits $\ceil{\log(s)}$ bits.
	
	One can put any repair scheme under this framework, and the difficulty of the problem stems from finding the functions $\mu_i$ that are informative enough, i.e., they provide enough information about $c_i$, from which collectively it is possible to compute the symbol $c_n$. However, they should not be too informative, in the sense that the size of the image, $s$, should be small to minimize the total incurred bandwidth. We have the following simple observation.
	\begin{obs}
		\label{obs1}
		A set of $n-1$ functions $\mu_i:\mathbb{F}_p\rightarrow [s]$ can be  extended to a repair scheme, i.e., there exists a function $G$ that satisfies \eqref{stamstam} if and only if for any two codewords $c,c'\in \cC$, such that $\mu_i(c_i)=\mu_i(c'_i)$ for all $i\in[n-1]$ it holds that $c_n=c'_n$.
	\end{obs}
	\begin{proof}
		Let $\mu_i$ for $i\in [n-1]$ be the $n-1$ functions as above, and define the function $G$ as follows. For a codeword $c\in \cC$ define the value of $G$ as in \eqref{stamstam}, i.e., $G(\mu_{1}(c_1),\ldots,\mu_{n-1}(c_{n-1}))=c_n.$ For all other points of $[s]^{n-1}$ define the value of $G$ arbitrarily. 
		By the property the functions $\mu_i$ satisfy, it is clear that the function $G$ is well defined, and together they form a valid repair scheme. 
		The other direction is trivial. 
	\end{proof}

	Every function $\mu_i$ defines a partition $\{\mu_i^{-1}(a): a\in [s]\}$ of $\Fp$.
	Vice versa, any partition of $\Fp$ to $s$ sets defines a function  whose value at the point $a\in \Fp$ is the index of the set that contains it. Hence, in the sequel, we will define the functions $\mu_i$ by  partitions of $\Fp$ to $s$ sets. 
	This paper's main contribution is identifying the  `right'  functions $\mu_i$, equivalently, the `right' partitions of $\mathbb{F}_p$ that define the $\mu_i$'s. As it turns out,  arithmetic progressions are the key for constructing the needed partitions that give rise  to an   
	efficient nonlinear repair schemes for codes over prime fields, as explained next. 
	
	Fix an integer $1\leq t\leq p$, set $s=\lceil  p/t \rceil$ and define $A_0,\ldots, A_{s-1}$ to be the partition of $\Fp$ into the following $s$ arithmetic progressions of length $t$ and step  $1$ 
	\begin{equation} \label{eq:Fp-partition}
	A_j= \begin{cases}
	\{ j t, j t + 1, \ldots, j t + t - 1\}  & 0\leq j\leq s-2\\
	\{(s-1)t, \ldots, p - 1\} & j=s-1. 
	\end{cases}
	\end{equation}
	For a nonzero  $\gamma \in \Fp$, it is easy to verify that $\gamma \cdot A_0,\ldots,\gamma \cdot A_{s-1}$ is also a partition of $\Fp$  into arithmetic progressions of length $t$ and step $\gamma$. Each function $\mu_i, i \in [n-1]$ of the repair scheme  will be defined by a partition $\gamma_i \cdot A_0,\ldots,\gamma_i \cdot A_{s-1}$
	for an appropriate selection of $\gamma_i$. Notice that the $\gamma_i$'s will be distinct for distinct $i$'s and therefore also the functions $\mu_i$ will be distinct for distinct $i$'s. Phrasing Observation \ref{obs1} in the language of partitions gives the following. The partitions defined by the $\gamma_i$'s extend to a repair scheme if (and only if) for any two codewords $c,c'\in \cC$  that belong to the same set in all of the $n-1$ different partitions, i.e., $c_i,c'_i\in \gamma_i \cdot A_{j_i}$ for all $i\in [n-1]$, it holds that $c,c'$ agree on their $n$th symbol, i.e., $c_n=c'_n$. 
	In such a case, we say that the $\gamma_i$'s define a valid repair scheme, and in the following proposition, we provide a relatively simple sufficient but instrumental condition for it. 
	\begin{prop} \label{con:repair-condition0}
		Let  $\cC\subseteq \Fp ^n$ be a linear code, $t < p$  be an integer,  and $\gamma_1,\ldots,\gamma_{n-1}$ be   nonzero elements of $\Fp$. If for any $c\in \cC$ with $c_i \in \gamma_{i} \cdot [-t,t]$ for all $1\leq i\leq n-1$, it holds that $c_n = 0$, then, the $\gamma_i$'s define  a valid repair scheme for the $n$th node with a total bandwidth of $(n-1)\log \left( \ceil{p/t}\right)$ bits.	
	\end{prop}
	
	\sloppy
	{Before proving the proposition, we remark that the actual number of bits each node sends is $\ceil{\log (\ceil{p/t})}$ and thus the total bandwidth is  $(n-1) \ceil{\log (\ceil{p/t})}\leq  (n-1) \log (\ceil{p/t}) + n - 1$. Since we think of $n$ as being a small constant compared to $p$, the additive factor of $(n-1)$  is negligible  when $\log (\ceil{p/t}) = \Omega (\log(p))$, which   is the case throughout the paper. Thus, for ease of notation we will omit the ceiling in the sequel, and assume that each node transmits $\log \left( \ceil{p/t}\right)$ bits.}
	
	\begin{proof}
		Let $s=\ceil{p/t}$ and define the sets $A_j, 0\leq j\leq s-1$ as in \eqref{eq:Fp-partition}. For each $i\in[n-1]$, define the function $\mu_i$ according to the partition $\gamma_i\cdot A_0,\ldots,\gamma_i\cdot A_{s-1}$. Namely, $\mu_i(a)=j$ if and only if $a\in \gamma_i\cdot  A_j$.
		Let $c, c'\in \cC$ be two codewords that agree on the $n-1$ values  $\mu_i(c)=\mu_i(c'), i\in [n-1]$. Then, $c-c'$ is a codeword such that its $i$th symbol belongs to the set $ \gamma_{i}\cdot [-t,t]$ for $i\in [n-1]$. Therefore, its $n$th symbol is equal to zero which implies that $c_n=c'_n$, and by Observation \ref{obs1} the $\gamma_i$'s define a valid repair scheme. The claim about the bandwidth follows since each partition consists exactly $s$ sets. 
	\end{proof}
	Since our primary focus is RS codes,  the following proposition is a   specialization of  \Cref{con:repair-condition0} to this case. In fact, we also slightly generalize it to address the case of arbitrary node repair and an arbitrary set of helper nodes. 
	\begin{prop}
		\label{con:repair-condition}
		Consider an  $[n,k]_p$ RS code defined by the  evaluation points  $\alpha_1, \ldots, \alpha_n$.  Let $\ell \in [n]$ be the failed node and $\mathcal{D}\subset [n]\backslash\{\ell \}$ be a set of $d$ helper nodes for $k\leq d\leq n-1$. Furthermore,   let $t<p$ be an integer and $\gamma_i, i\in \mathcal{D}$ be   nonzero elements of $\Fp$. 
		If for any polynomial $f(x)\in \Fp[x]$ of degree less than $k$ with  $f(\alpha_i)\in \gamma_i\cdot [-t,t]$ for all $i\in \mathcal{D}$, it holds that $f(\alpha_{\ell})=0$, then, the $\gamma_i$'s define  a valid repair scheme for the $\ell$th node with a total bandwidth of $d\log \left( \ceil{p/t}\right)$ bits.
	\end{prop}
	
	\begin{proof}
		The result is obtained by applying  \Cref{con:repair-condition0} to the punctured $[d+1,k]_p$ RS code defined by  the evaluation points $\{ \alpha_i \mid i\in \mathcal{D} \cup \{\ell\}\}$.
	\end{proof}
	All the efficient repair schemes given in the paper will follow by showing that there is a choice of evaluation points $\alpha_1, \ldots, \alpha_n\in \Fp$ such that any node (and not only the last node) can be efficiently repaired by invoking \Cref{con:repair-condition} with carefully designed partitions ($\gamma_i$'s).
	
	We say that  a polynomial $f\in \Fp[x]$ \emph{passes through} $(\alpha, A)$    for $\alpha \in \Fp$ and a subset $A\subseteq \mathbb{F}_p$ if $f(\alpha)\in A$.  
	Figure \ref{fig:repair_map_graph} illustrates a valid repair scheme (defined by the $\gamma_i$'s) of  the $n$th node of a RS code, that satisfies the condition of Proposition \ref{con:repair-condition}. Namely, any two polynomials $f(x),g(x)$  that pass through the same set in each of the $n-1$ partitions,  attain the same value at $\alpha_n$, i.e., $f(\alpha_n)=g(\alpha_n)$. 
	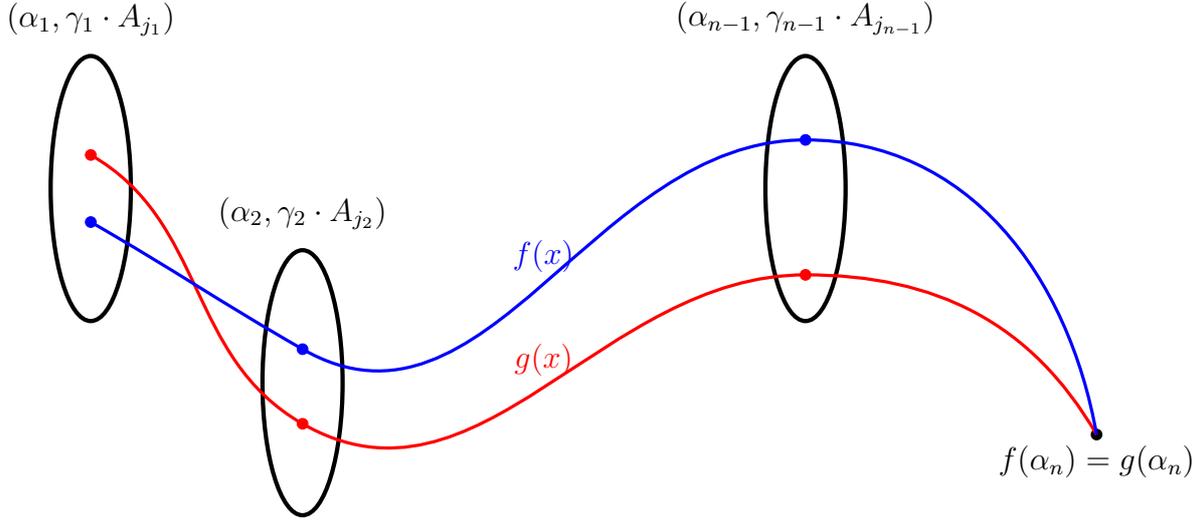
\begin{figure}[h]
		\begin{tikzpicture}
		
		\usetikzlibrary{arrows}
		\usetikzlibrary{shapes}
		\usetikzlibrary{positioning}
		
		\def\setnodeone#1#2{
			\node[ultra thick, draw=black, ellipse,minimum height=100pt, minimum width=30pt, align=center] (#1) {#2}	
		}
		\def\setnodetwo#1#2#3{
			\node[ultra thick, draw=black, ellipse,minimum height=100pt, minimum width=30pt, align=center,below right=0.05cm and 2cm of #3] (#1) {#2}
		}
		\def\lastnode#1#2#3{
			\node[ultra thick, draw=black, ellipse,minimum height=100pt, minimum width=30pt, align=center,right=2cm of #3] (#1) {#2};
		}
		
		\setnodeone{a}{};
		\node[above=0.1cm of a] (set1) {$(\alpha_1, \gamma_1 \cdot A_{j_1})$};
		\node[above=0.3cm of a.center] (a1) {};
		\filldraw[red] (a1.center) circle (2pt);
		\node[below=0.3cm of a.center] (a2) {};
		\filldraw[blue] (a2.center) circle (2pt);

		\setnodetwo{b}{}{a};
		\node[above=0.1cm of b] (set1) {$(\alpha_2, \gamma_2 \cdot A_{j_2})$};
		\node[above=0.3cm of b.center] (b2) {};
		\filldraw[blue] (b2.center) circle (2pt);
		\node[below=0.4cm of b.center] (b1) {};
		\filldraw[red] (b1.center) circle (2pt);
		
		\node[right=6cm of a] (dots) {};
		
		\lastnode{last}{}{dots};
		\node[above=0.1cm of last] (set1) {$(\alpha_{n-1}, \gamma_{n-1} \cdot A_{j_{n-1}})$};
		\node[above=0.5cm of last.center] (l2) {};
		\filldraw[blue] (l2.center) circle (2pt);
		\node[below=1cm of last.center] (l1) {};
		\filldraw[red] (l1.center) circle (2pt);

		\node[below right=2cm and 2 cm of last] (po) {$f(\alpha_n) = g(\alpha_n)$} ;
		\filldraw[black] (po.north) circle (2pt);

		\draw[red,very thick] (a1.center) to[out=-30,in=150] (b1.center);
		\draw[red,very thick] (b1.center) to[out=-30,in=180] node[above] {$g(x)$} (l1.center);
		\draw[red,very thick] (l1.center) to[out=0,in=120] (po.north);
		
		\draw[blue,very thick] (a2.center) to[out=-30,in=150] (b2.center);
		\draw[blue,very thick] (b2.center) to[out=-30,in=180] node[above] {$f(x)$} (l2.center);
		\draw[blue,very thick] (l2.center) to[out=0,in=100] (po.north);
		\end{tikzpicture}
		\caption{A valid repair of the $n$th node of an $[n,k]$ RS code, that satisfies the condition of Proposition  \ref{con:repair-condition}. The   two polynomials  $f(x)$ and $g(x)$ of degree less than $k$    pass through $(\alpha_i,\gamma_i\cdot A_{j_i})$ for every $i \in[n-1]$ and hence $f(\alpha_n) = g(\alpha_n)$.}
		\label{fig:repair_map_graph}
	\end{figure}
	
	
	

	\section{Asymptotically MSR RS codes over $\Fp$}
	\label{sec:exsitence-results}
	In this section, we show the existence of RS codes over prime fields that have efficient repair schemes, but first  
	we begin with the following two definitions that will be used in the sequel.
	
	\begin{defi}
		A code symbol of an $[n,k]$ (scalar) MDS code over a field $\mathbb{F}$ is  said to admit an asymptotically optimal repair (bandwidth) if  it can be  repaired by some set of  $d$  helper nodes (for $k\leq d<n$) and bandwidth at  most 
		\begin{equation}
		\label{asmptotic-cut-set}
		\log(|\mathbb{F}|)\left( \frac{ d}{d-k+1}+o(1) \right)
		\end{equation}
		bits, where the term $o(1)$ tends to zero as the field size $|\mathbb{F}|$ tends to infinity. 
	\end{defi}
	
	\begin{defi}
		An $[n,k]$ (scalar) MDS code  is said to be  asymptotically $[n,k,d]$ MSR if any of its code symbols admits asymptotically optimal repair, i.e., satisfy \eqref{asmptotic-cut-set}, by \emph{any} set of $d$ helper nodes.     
	\end{defi}
	We proceed to show that over large enough prime fields, there exist $[n,2]$ RS codes that are asymptotically $[n,2,d]$ MSR codes for every $2<d<n$. Then, we proceed to generalize the result to RS over field extensions. 
	We would like to stress  that the constructions and those presented in the following sections are  asymptotically optimal in the strong sense. Namely, the actual bandwidth differs from the cut-set bound by an additive constant that depends only on the parameters $n,k,d$ and not on the alphabet size.  
	Furthermore, the constructions' repair schemes outperform all the linear repair schemes, which is  a phenomenon that  was not known to exist before.  In particular, for prime fields, the only known repair scheme is the trivial scheme (which is a linear scheme).
	Therefore, by outperforming the trivial repair over prime fields, we obtain the first known example of a nontrivial repair over prime fields. 
	
	\subsection{Existence of asymptotically  $[n,2,d]$ MSR RS codes over $\mathbb{F}_p$} \label{sec:n-2-existence}
	
	In this section we show by a counting (encoding)  argument the existence of an asymptotically $[n,2,d]_p$ MSR RS code for a large enough prime $p$ and any $2\leq d\leq n-1$. In fact, we show a stronger result, that is,  for any $\varepsilon>0$ and a large enough prime $p$, a fraction of $1-\varepsilon$ of  all the   RS codes  satisfy this property. 
	In particular, for  such a code 
	every code symbol $i\in [n]$ can be repaired by any  $d$  helper nodes, where each helper node transmits $(1/(d-1)) \log(p) + O_{n,\varepsilon}(1)$ bits, which is roughly a $1/(d-1)$ fraction of the amount of information it holds. 
	\begin{thm} \label{thm:2-n-case}
		Let  $\varepsilon>0$ and  $2\leq d < n$, then for  a large enough prime $p$, a fraction of $1-\varepsilon$ of  all the $[n,2]_p$  RS codes
		are asymptotically $[n,2,d]$ MSR  codes with repair bandwidth  $\frac{d}{d-1} \log(p)  +O_{n,\varepsilon}(1).$
	\end{thm}
	
	\begin{proof}
		Let $t<p$ be an integer to be determined later, and let $\alpha_1,\ldots, \alpha_n$ be the evaluation set of the $[n,2]_p$ RS code. We would like to show that any code symbol    admits an asymptotically optimal repair by applying Proposition \ref{con:repair-condition}.
		Assume that the $\ell$th symbol has failed  and that the there exists a set $\mathcal{D}\subseteq [n]\backslash \{\ell\},|\mathcal{D}|=d$ that   does not satisfy the condition in Proposition \ref{con:repair-condition} with $\gamma_i:= \alpha_i- \alpha_\ell$
		, for $i\in \mathcal{D}$. Therefore,  there exists a polynomial $f(x)$ of degree at most one   that passes through $(\alpha_i, \gamma_i\cdot [-t,t])$ for $i\in \mathcal{D}$ and  $f(\alpha_\ell)\neq 0$.
		Let $j=\min\{ \mathcal{D}\} $ and define the polynomial 
		\[
		g(x) := f(x) - \frac{f(\alpha_j)}{\gamma_j} (x-\alpha_\ell).\;
		\]
		Notice that   $g(\alpha_j) = 0$ and $g(\alpha_\ell)=f(\alpha_\ell)\neq 0$, hence $g(x)$ is of the form $g(x) = m(x - \alpha_{j})$ for some $m\neq 0$. 
		Also, for $i\in \mathcal{D}\backslash \{j\}$
		\begin{equation}
		\label{main-thm}
		0\neq g(\alpha_i)=
		f(\alpha_i) - \frac{f(\alpha_j)}{\gamma_j} \gamma_i\in  \gamma_i\cdot [-2t,2t],\end{equation}
		since $f(\alpha_i), \frac{f(\alpha_j)}{\gamma_j} \gamma_i\in  \gamma_i\cdot [-t,t]$. 
		Here, the fact that  $\gamma_i\cdot [-t,t]$ is an arithmetic progression plays a crucial rule, since it implies that the size of the set  $\gamma_i\cdot [-t,t]-\gamma_i\cdot [-t,t]$ is small.  
		We conclude that one can find a linear polynomial $g(x)$ with the above properties for any such evaluation set. 
		Next, we give an encoding for  the (bad) evaluation sets.
		
		\vspace{0.3 cm}{\bf Encoding: }\begin{enumerate}
			\item Encode the index of the failed symbol $\ell$, and the set $\mathcal{D}.$ There are $n\binom{n-1}{d}$ options for this.
			\item Encode  the evaluation points
			$\alpha_i,i\notin \mathcal{D}\cup \{\ell\}$. There are $p^{n-d-1}$ options for this. 
			\item Encode the evaluation points $\alpha_j,\alpha_\ell$ and $\alpha_k$, where $k$ is the second smallest element of $\mathcal{D}$. There are at most $p^3$ options for this.
			\item Encode the value $g(\alpha_k)$. By \eqref{main-thm} there are at most $4t$ options for this.
			\item For each $i\in \mathcal{D}\backslash\{j,k\}$ encode the value $g(\alpha_i)/\gamma_i$ which is in the set $[-2t,2t]\backslash\{0\}$  by \eqref{main-thm}. Again, there at most $(4t)^{d-2}$ for this. 
		\end{enumerate}
		Next, we show that given the encoding, one can recover the original evaluation set. In other words, the encoding mapping is injective. 
		
		\vspace{0.2 cm}
		{\bf Decoding:}
		Given Steps (1)-(2) one can recover the index of the failed symbol, the set $\mathcal{D}$ and the evaluation points $\alpha_i,i\notin \mathcal{D}\cup \{\ell\}$. 
		Given Steps (3)-(4) and the fact that $g(x) $ is a nonzero polynomial of degree at most one that vanishes at $\alpha_j$, we can find the value of  $m$ and hence recover  the polynomial $g(x) = m(x - \alpha_j)$.
		Next, it remains to recover the points $\alpha_i$, for $i\in \mathcal{D}\backslash\{j,k\}$. Given the value of $\alpha_{\ell}$ and  $g(\alpha_i)/\gamma_i$  obtained in Steps (4) and (5), respectively,  one can construct the \emph{non-degenerate} equation  $$g(\alpha_i)-(\alpha_i-\alpha_\ell)\cdot \frac{g(\alpha_i)}{\gamma_i}=0,$$
		in the variable $\alpha_i$. Since it is a linear equation,  the value of $\alpha_i$ can be uniquely determined.
		
		Hence,  the number of such (bad) evaluation sets, i.e., that do not satisfy the condition  in Proposition \ref{con:repair-condition} for repairing any of the $n$ symbols is at most
		the total number of possible encodings, which is at most  $n\binom{n-1}{d} p^{n-d+2}(4t)^{d-1}$ options.
		Set 
		\[
		t=\ceil{ \frac{\varepsilon p^{\frac{d-2}{d-1}}}{10  n^{\frac{d+1}{d-1}}} },
		\]
		and  then $n\binom{n-1}{d} p^{n-d+2}(4t)^{d-1} < \frac{\varepsilon}{2} \cdot p^n$. On the other hand,  the number of possibilities to choose the $ \alpha_i$'s is $p^n (1 - o(1))$ where  the term $o(1)$ tends to zero as $p$ tends to infinity. This implies that at least  $p^n(1-\frac{\varepsilon}{2} -o(1))$  of the evaluation sets are not bad. Hence,  for large enough $p$  a fraction of at least $1-\varepsilon$  of the possible RS codes satisfy the condition in Proposition \ref{con:repair-condition}. 
		Lastly, for such an RS code, the total incurred bandwidth during the repair of any symbol is
		\[
		d \cdot \log\left( \ceil{ \frac{p}{t}} \right) \leq d \cdot \log \left( p^{\frac{1}{d-1}} \cdot 10 \varepsilon^{-1}  \cdot n^{\frac{d+1}{d-1}} \right) = \frac{d}{d-1} \log(p)  +O_{n,\varepsilon}(1), \;
		\]		
		bits\footnote{{As discussed above, the actual bandwidth is $d \cdot \ceil{\log\left( \ceil{ \frac{p}{t}} \right)}$, but the outer ceiling adds at most $d$ bits, which are absorbed in the $O_{n,\varepsilon}(1)$ term, as $d<n$. Hence,  it does not affect the final bandwidth.}}, as needed. 
	\end{proof}

	\subsection{Outperforming linear repair schemes  over field extensions}
	In Section \ref{sec:n-2-existence}  we obtained the first  existence result of an RS code that can be asymptotically repaired over a prime field. On the other hand,  any linear repair scheme over a prime field is the trivial repair, i.e., repairing with any $k$ helper nodes that transmit their entire symbol, where $k$ is the dimension of the code. Indeed, prime fields have no proper subfields. The only option for  a helper node is to transmit its entire symbol. Therefore, the previous section's result can be viewed as the first example wherein a nonlinear repair scheme outperforms all linear ones.

	A natural question to consider is whether this phenomenon could be found over field extensions or it is solely a property of prime fields. In this section, we show by relying on the result of Section \ref{sec:n-2-existence} and a simple extension of  repair schemes over some field to   repair schemes over  its field extensions, that this phenomenon also occurs over field extensions. More precisely, we exhibit the existence of RS codes over $\mathbb{F}_{p^m}$ for infinitely many primes $p$ and integers $m$ that have a nonlinear repair scheme that outperforms (in terms of the total incurred bandwidth) any linear repair scheme.   
	The outperformance of the nonlinear scheme over linear ones follows from the fact
	that the latter requires the transmission of $\Fp$-symobls, whereas  there is no such constraint on the former.
	
	The following  proposition  shows that an RS code over  a  field $\mathbb{F}$ and   evaluation points in a subfield of $\mathbb{F}$ can be repaired by invoking any repair scheme of the  RS code  over the subfield and the same evaluation set. In other words, a repair scheme for an $[n,k]_p$ RS code can be  translated to  a repair scheme for an $[n,k]_{p^m}$ RS code. We note that a similar result was already given   in \cite[Theorem 1]{li2019sub}, but for the sake of  completeness, we state and prove it here again.
	
	\begin{prop}
		\label{pakapaka}
		Let $\alpha_1, \ldots, \alpha_n\in \Fp$ be the evaluation set of an $[n,k]_p$ RS code, and suppose that there exists a repair scheme for node $i$ with  a set $D\subset [n]$ of $d$ helper nodes and  bandwidth of $b$ bits. 
		Then, for every  positive integer $m$  the  $[n,k]_{p^m}$ RS code with the same evaluation set $\alpha_1, \ldots, \alpha_n\in \Fp$   
		has a repair scheme for node $i$ with  the same set $D\subset [n]$ of $d$ helper nodes and  bandwidth of $b\cdot m$ bits.

		
	\end{prop}
	\begin{proof}
		Let $\beta_0,\ldots,\beta_{m-1}\in \mathbb{F}_{p^m}$ be a basis of $\mathbb{F}_{p^m}$ over  $\Fp$.
		It can be readily verified that any polynomial $f(x)\in \mathbb{F}_{p^m}[x]$ can be written as
		$$f(x)=\sum_{j=0}^{m-1}f_j(x)\beta_j, \text{ where } f_j(x)\in \Fp[x].$$
		Then, the problem of repairing the value $f(\alpha_i)$ for a polynomial $f(x)\in \mathbb{F}_{p^m}[x]$ of degree less than $k$ boils down to $m$ independent repairs of   $f_j(\alpha_i)$ for $j=0\ldots,m-1$ over $\Fp$. By invoking $m$ times the repair scheme for $\alpha_j,j=0,\ldots,m-1$, we get that the total incurred bandwidth is $b\cdot m$, as needed. 
	\end{proof}

	\begin{prop}
		There are infinitely many primes $p$ and positive integers $m$, and $2\leq d<n$ for which there exist an $[n,2]_{p^m}$ RS code  with a nonlinear repair scheme with $d$ helper nodes  that outperforms any linear repair scheme. 
	\end{prop}
	
	\begin{proof}
		Fix an $\varepsilon>0$, and  consider the $[n,2]_p$ RS code with evaluation set $\alpha_1,\ldots,\alpha_n$, given in Theorem \ref{thm:2-n-case}.  Let $m$ be a positive integer not divisible by $d-1$, and consider the $[n,2]_{p^m}$ RS code
		with the same evaluation set. By Proposition \ref{pakapaka} there exists a nonlinear repair scheme for this code with bandwidth at most 
		$$m\left(\frac{d}{d-1} \log(p)  +O_{n,\varepsilon}(1)\right).$$
		On the other hand, any linear repair scheme over $\Fp$ requires the transmission of $\Fp$-symbols, then by the cut-set bound \eqref{eq:cut-set} the number of $\Fp$-symbols  transmitted is at least  $\ceil{dm/(d-1)}$. Since $d-1 \nmid m$ then any linear repair scheme transmits at least 
		\begin{equation}
		\label{difference}
		\left( \ceil{\frac{dm}{d-1}}-\frac{dm}{d-1}\right) \log (p)-O_{m,n,\varepsilon}(1),
		\end{equation}
		more bits than the nonlinear repair scheme, as needed. 
	\end{proof}
	Although the improvement of the nonlinear repair scheme over  linear schemes is small and might seem very negligible, the difference \eqref{difference} can be arbitrarily large as we increase the alphabet size $p$. 
	The purpose of this result is to exemplify that nonlinear repair schemes can, in fact, outperform linear schemes even over field extensions and not only over prime fields. We believe that it is possible to exhibit even a more significant gap between the two.

	\section{Explicit constructions of RS codes}
	\label{stam-section}
	In the  previous section, we showed the   existence of $[n,2]_p$ RS codes that are asymptotically MSR. In this section, we turn our focus to explicit constructions of RS codes that have efficient repair schemes.  We present several such constructions by explicitly presenting the evaluation points $\alpha_i$ and the $\gamma_i$'s that define the partitions (functions) computed by the helper nodes, as detailed below.
	
	We begin  with a toy example of an asymptotically $[4,2,3]_p$ MSR RS code. The result will follow  by showing that the explicit $\gamma_i$'s and the evaluation points $\alpha_i$ satisfy the condition of  \Cref{con:repair-condition}.  Then, we continue to present the main construction of this section. Building on the ideas presented in the toy example, we explicitly construct for all $k < d \leq n/2$ an $[n,k]_p$ RS code such that every node admits an asymptotically optimal repair   with many $d$-sets of nodes as helper nodes, although not all of them. Therefore,  the code  falls short of being  asymptotically MSR. 
	We conclude the section with  two more explicit constructions of full-length and folded RS codes that follow  directly from this section's main construction. 
	\subsection{A toy example}
	Assume that we would like to construct a $[4,2]_p$ RS code with asymptotically optimal repair	for the $4$th symbol by invoking Proposition \ref{con:repair-condition}. Then, we need to choose distinct points $\alpha_1, 
	\ldots, \alpha_4\in \Fp$ and $\gamma_1, \ldots, \gamma_3\in \Fp$ for which  Proposition \ref{con:repair-condition} holds for $t = \Omega(\sqrt{p})$. Indeed, for such a $t$, each helper nodes transmits $\frac{1}{2}\log(p)+O(1)$ bits. 
	
	
	Let $p$ be a large enough prime, and let   $\gamma_i :=\alpha_i-\alpha_4$ for $i=1,2,3$, where  
	\[
	\alpha_1 = 0, \alpha_2 = -1, \alpha_3 = \frac{p-1}{2}, \alpha_4 = -(2t + 1)\;,
	\]
	and  $t= \floor{\sqrt{p}/5}$. Then,  by Proposition \ref{con:repair-condition}  the repair of the $4$th symbol is possible if for any polynomial $f\in \Fp[x]$ of degree at most one,  that satisfies $f(\alpha_i)\in \gamma_i\cdot [-t,t]$ for $i=1,2,3$, it holds that $f(\alpha_4)=0$. Let $f$ be such a polynomial, i.e., $\deg(f)\leq 1$ and $f(\alpha_i)\in \gamma_i\cdot [-t,t]$ for $i=1,2,3$, write $f(\alpha_1)=m\cdot \gamma_1$ for some $m\in [-t,t]$, and consider the polynomial $\hat{f}:=m(x-\alpha_4). $  We would like to show that by the above  selection of the $\alpha_i$'s, $f$ is equal to $\hat{f}$ and  in particular $f(\alpha_4)=\hat{f}(\alpha_4)=0$, as needed. 
	Note that $f(\alpha_i),\hat{f}(\alpha_i)\in \gamma_i\cdot  [-t,t]$ for $i=1,2,3$ and $f(\alpha_1)=\hat{f}(\alpha_1)$, then their difference $g:=f-\hat{f}$ satisfies, $\deg(g)\leq 1$, $g(\alpha_1)=0$. Therefore, $g(x)=s(x-\alpha_1)$ for some $s\in \mathbb{F}_p$, and $g(\alpha_i)\in \gamma_i\cdot [-2t,2t]$ for $i=2,3$. To conclude, we will show that $s=0$ and therefore $g$ is the zero polynomial. 
	
	Towards this end, notice first that $s=\frac{g(\alpha_i)}{\alpha_i-\alpha_1}  \in  \frac{\gamma_i}{\alpha_i-\alpha_1} \cdot[-2t,2t]$, for $i=2,3$, therefore 
	\begin{equation}
	\label{wsws}
	s\in \bigcap_{i=2,3}  
	\frac{\gamma_i}{\alpha_i-\alpha_1}\cdot[-2t,2t].
	\end{equation}
	Next, assume that we take a realization of $\mathbb{F}_p$ as all the integers whose absolute value is less than $p/2$, i.e., $F_p=\{0,\pm 1,\pm 2,\ldots, \pm \frac{p-1}{2}\}$, then 
	$$ \frac{\gamma_2}{\alpha_2-\alpha_1}=-2t \text{ and }  \frac{\gamma_3}{\alpha_3-\alpha_1}=\frac{\alpha_3-\alpha_4}{\alpha_3-\alpha_1}=1-\frac{\alpha_4}{\alpha_3}=1-2(2t+1)=-4t-1.
	$$
	Therefore, the following products (over $\mathbb{Z}$) satisfy 
	$$\left|2t\cdot\frac{\gamma_2}{\alpha_2-\alpha_1}\right|,\left|2t\cdot\frac{\gamma_3}{\alpha_3-\alpha_1}\right|<\frac{p}{2}.$$
	This implies that  
	there is no wrap around in the calculation of the sets $\gamma_i/(\alpha_i-\alpha_1)\cdot[-2t,2t]$  in \eqref{wsws}   when viewed as  subsets of $\mathbb{F}_p$, and they are equal to the same sets of products when calculated over $\mathbb{Z.}$
	
	By \eqref{wsws}, $s$ is divisible by $-4t-1$ and by $2t$, and thus, $s$ is divisible by their lcm. Since $-4t-1$ and $-2t$ are coprime, then  
	\begin{equation}
	\label{wsws1}
	\lcm(-4t-1,-2t)= 2t(4t+1)>2t\cdot \left|\frac{\gamma_2}{\alpha_2-\alpha_1}\right|=4t^2.
	\end{equation}
	Now assume to the contrary that $s\neq 0$, then by \eqref{wsws1}, $s$  is a nonzero integer whose absolute value is greater than $4t^2$. We get that the absolute value of $s$ is greater than the maximal absolute value of any element in the set $\frac{\gamma_2}{\alpha_2-\alpha_1}\cdot [-2t,2t]$, and we arrive at a contradiction (since \eqref{wsws} does not hold).
	%
	Therefore, the  repair of $\alpha_4$ is possible, where each helper nodes transmits at most 
	$$\log \left(\ceil{\frac{p}{t}}\right)\leq \frac{1}{2}\log(p)+O(1),$$
	bits, as needed. 
	
	In \cref{sec:repairing-rest} we  show that  by applying again Proposition \ref{con:repair-condition}, the other evaluation points admit an asymptotically optimal repair with $d=3$. Therefore this is an asymptotically $[4,2,3]$ MSR code. 
	
	\subsection{An RS code construction with $k < d \leq n/2$}
	This section presents an explicit construction of RS codes over $\Fp$ with efficient repair, as described next. Let  $d,k$, and $n$ be arbitrary positive integers such that $k< d\leq n/2$. We construct an $[n,k]_p$ RS code with the following repair properties. The set of $n$ nodes can be partitioned into equally sized sets,   of size $n/2$ (for simplicity, assume that $n$ is even, for odd $n$ the size of the sets differ by one), such that any failed node can be optimally repaired by any subset of $d$ helper nodes from the other set, i.e., the set that does not contain the failed node. 
	Clearly, this constraint  provides for each failed node   $\binom{n/2}{d}$ possible helper sets; however, the construction falls short in satisfying the definition of an asymptotically MSR code since not any set of $d$ nodes can serve as a set of helper nodes.
	Another caveat of this construction that we should mention is its low rate, as $k/n\leq 1/2.$ It is an interesting open question whether it is possible to modify this construction to a code without these constraints. 
	
	Before stating and proving the main result of this section, we  state a lemma that provides a lower bound on the least common multiple of  several integers.
	\begin{lemma} \label{lem:lcm-lemma}
		Let $a_1, \ldots, a_{s}$ be positive integers, then 
		\[
		\lcm(a_1, \ldots, a_{s}) \geq \frac{\prod_{i=1}^s a_i}{\prod_{1\leq i < j\leq s} \gcd (a_i, a_j)} \;.
		\]
	\end{lemma}
	The proof of the Lemma \ref{lem:lcm-lemma} is given in the Appendix (Section \ref{proof-of-the-lemma}). Next, we present the explicit construction of the RS code with the claimed properties. 

	\begin{cnst} 
		\label{cnst}
		Let $k,d,n$ be  integers such that $k < d \leq n/2$ and $n$ is even.  
		Let $r := \lfloor p^{\frac{1}{d-k+1}} \rfloor$, where $p$ is a large enough prime, and 
		define the $[n,k]_p$ RS with $n$ evaluation points $\alpha_i$ where 
		$$\alpha_i=\begin{cases}
		i& i\in [n/2]\\
		r+i-\frac{n}{2}&  i\in [n/2+1,n].
		\end{cases}$$
	\end{cnst}
	
	The following theorem shows that the constructed RS code admits an asymptotically optimal repair for any of its nodes (symbols). 
	\begin{thm}
		\label{thm:n-k-d-allmost-msr}
		The $[n,k]_p$ RS code given in Construction \ref{cnst} admits a partition of its nodes to two sets $[n/2]$ and $[n/2 + 1,n]$ such that any node from one set can be repaired by  any set of $d$ helper nodes from the other set, with  total bandwidth  of  $d/(d-k+1) \log(p) + O_{n,k,d}(1)$ bits. 
	\end{thm}

	\begin{proof}
		The result will follow by showing that for carefully designed $\gamma_i$'s, \Cref{con:repair-condition} holds true.
		Let 
		$\delta \in [r + 1,r + n/2]$ be the evaluation point of the failed node and note that the proof for the other  evaluation points  is almost identical and thus is omitted.
		Let $\alpha_1,\ldots,\alpha_d \in [n/2]$ be a set of $d$ distinct helper nodes, and define the $\gamma_i$'s as

		$$
		\gamma_i=
		\begin{cases}
		(-1)^k (\alpha_i - \delta) \prod_{\substack{ j=1\\ j\neq i}}^d (\alpha_j - \alpha_i)  &  1\leq i \leq k-1\\

		(\alpha_i - \delta) \prod_{j=1}^{k-1} (\alpha_i - \alpha_j)  & k\leq i\leq d.
	\end{cases}
	$$
	
	Consider   a polynomial $f(x)$ of degree less than $ k$ such that $f(\alpha_j) := \beta_j \in \gamma_j \cdot T$ for all $1\leq j \leq d$,  where $T = [-t,t]$ and $t$ is a positive integer to be determined later. We will show that in such a case   $f(\delta) = 0$ 
	and hence  \Cref{con:repair-condition} holds. Let $h(x)$  be the polynomial of degree less than $k$ defined by the $k$ constraints,  $h(\delta) = 0$ and $h(\alpha_i) = f(\alpha_i)= \beta_i$ for all $1\leq i \leq k-1$. By the Lagrange interpolation formula one can easily verify that $h(x)$ takes the following form
	\[
	h(x) = \sum_{i=1}^{k-1} \frac{x-\delta}{\alpha_i - \delta} \prod_{\substack{j=1\\ j\neq i}}^{k-1} \frac{x - \alpha_j}{\alpha_i - \alpha_j}\beta_i.
	\]
	Next, define $g(x) := f(x) - h(x)$ and note that since $g(x)$ is a polynomial of degree less than $k$  that vanishes at $\alpha_j$ for all $1\leq j \leq k-1$,  it takes the following form 
	\begin{equation} \label{eq:g-def}
	g(x) = a \prod_{j = 1}^{k - 1}(x - \alpha_j),\;
	\end{equation}
	for some $a\in \mathbb{F}_p$. We wish to show that $a = 0$ which implies that $g \equiv 0$ and $f(\delta) = 0$, as needed.
	For $m\in [k, d]$ 
	\begin{equation*}
	g(\alpha_m) = f(\alpha_m) - h(\alpha_m) 
	= \beta_m - \sum_{i=1}^{k-1} \frac{\alpha_m-\delta}{\alpha_i - \delta} \prod_{\substack{j=1\\ j\neq i}}^{k-1} \frac{\alpha_m - \alpha_j }{\alpha_i - \alpha_j } \beta_i \;.
	\end{equation*}
	Combined with \eqref{eq:g-def} we have
	\[
	a = \frac{\beta_m}{\prod_{j = 1}^{k - 1}(\alpha_m - \alpha_j)} - \sum_{i=1}^{k-1} \frac{\alpha_m-\delta}{\alpha_i - \delta}\frac{1}{\alpha_m - \alpha_i} \prod_{\substack{j=1\\ j\neq i}}^{k-1} \frac{1}{\alpha_i - \alpha_j} \beta_i \;.
	\]
	Since $\beta_j  \in \gamma_j \cdot T$ for $j=1,\ldots, d$ then 
	\[
	a \in (\alpha_m - \delta)\cdot  T - (\alpha_m - \delta)  \sum_{i=1}^{k-1} \prod_{\substack{j=k\\ j\neq m}}^{d} (\alpha_j - \alpha_i)\cdot  T\;,
	\]
	i.e., $a$ belongs to a set which is a sum of $k$ sets. Equivalently we can write 
	\begin{equation}
	\label{set}
	a \in (\alpha_m - \delta)\cdot \left( T - \sum_{i=1}^{k-1} \prod_{\substack{j=k\\ j\neq m}}^{d} (\alpha_j - \alpha_i)\cdot T \right) \;.
	\end{equation}
	Consider a realization of $\Fp$
	as all the integers whose absolute value is less than $p/2$, i.e., $F_p=\{0,\pm 1,\pm 2,\ldots, \pm \frac{p-1}{2}\}$, and consider the set in \eqref{set}
	as a subset of $\mathbb{Z}$, where all the multiplications and additions are done over $\mathbb{Z}$. We claim that  the absolute value of any of its elements is at most $|\alpha_m-\delta|tk(n/2)^{d-k}$. 
	Indeed, $\alpha_j \in [n/2]$ for each $j=1,\ldots, d$ and therefore $\left|a_{j} - a_{i}\right| \leq n/2$. We conclude that the set in \eqref{set} is contained in the set 
	$$
	(a_m - \delta)\cdot [- Ct, Ct],
	$$
	where $C = C_{n,k,d}$ is a positive constant that depends only on $n,k$ and $d$.  Let $t = \xi p^{1 - \frac{1}{d - k + 1}} $, where $\xi$ is a small positive constant which will be determined later. 
	Next, if $\xi$ is small enough, then for any $m=k,\ldots, d$ the absolute value of an element of the set $(a_m - \delta)\cdot  [- Ct, Ct]$ is less than $p/2$, and therefore there is no wrap around when it is viewed as an element of $\Fp.$ To conclude, by \eqref{set} 		
	the integer $a$ is divisible by $\alpha_{m} - \delta$ for $m=k,\ldots, d$, hence  it is divisible by their least common multiple. By \Cref{lem:lcm-lemma}, and the fact that for distinct $\alpha_i, \alpha_j \in [n/2]$, it holds that $\gcd(\alpha_i-\delta, \alpha_j-\delta)=\gcd(\alpha_i-\delta, \alpha_j-\alpha_i) \leq n/2$, then 
	\[
	\lcm (\alpha_{k} - \delta, \ldots, \alpha_{d} - \delta) \geq \frac{|(\alpha_{k} - \delta) \cdots(\alpha_{d} - \delta)|}{\left(\frac{n}{2}\right)^{\binom{d-k+1}{2}}}=\Omega(\delta^{d - k + 1})=\Omega(r^{d - k + 1}) = \Omega (p),
	\]
	since $k,n$ and $d$ are constants with respect to $p$. 
	Thus, there exists a positive constant $\varepsilon = \varepsilon_{n,k,d} < 1$ such that $|a| > \varepsilon p $ or $a = 0$. 
	On the other hand, recall that $a \in (\alpha_m - \delta) \cdot [- Ct, Ct]$ and thus, for small enough $\xi$ we get that 
	$\left| a \right| < \varepsilon p$. 
	Thus, it must be that $a = 0$, which implies that $g\equiv 0$ and in particular $g(\delta)=f(\delta) = 0$. Therefore, by \Cref{con:repair-condition}, it is possible to repair the failed symbol node $\delta$, as needed. 
	The total incurred bandwidth by the scheme is
	$ d\cdot  \log(p/t) = d\log(p^{1/(d-k+1)}) + d \log(1/\xi)$ and the results follows immediately since $\xi$ depends only on $n,k$, and $d$.
\end{proof} 
Next, we shall make use of \Cref{cnst} to construct a  full-length RS code over $\Fp$ that has good repair properties. Namely, we show that every code symbol has at least $\Omega(p)$ distinct helper  sets  that enable an asymptotically  optimal repair.
\begin{thm}
	\label{thm:full-length-rs}
	Let $d > k$ be positive integers. Let $p$ be a large enough prime and  consider the full length $[p,k]_p$ RS code, Then, every failed code symbol has  $\Omega(p)$ distinct sets of helper nodes of size $d$ that can repair it with bandwidth at most $d/(d-k+1) \log(p) + O_{k,d}(1)$.
\end{thm}
\begin{proof}
	Let $G=\{ax+b:a,b\in \Fp, a\neq 0\}$ be the affine general linear group acting on $\Fp.$ It is well-known that $G$ is sharply $2$-transitive, and therefore  the subgroup $G_\delta=\{ax+b\in G: a\delta +b=\delta\}$ that stabilizes a point $\delta \in \Fp$ is sharply transitive on $\Fp\backslash \{\delta\}.$  
	
	Let $\delta\in \Fp$ be the failed node. Let $h\in G$ be an affine transformation such that $h(1)=\delta$ and set   $A:=h([r+1,r+d])=\{h(a):a\in [r+1, r+d]\}$, where 
	$r:= \lfloor p^{\frac{1}{d-k+1}} \rfloor$.
	We claim that $G_\delta^A=\{g(A):g\in G_{\delta}\}$, the orbit of the set $A$ under the action
	of $G_{ \delta}$ is  of size  $\Omega(p)$, and each set in the orbit forms a set of helper nodes for repairing the failed node $\delta$. 
	Indeed, it is well-known that the size of the orbit satisfies 
	$$|G_\delta^A|=\frac{|G_\delta|}{|G_{\delta,A}|}=\frac{p-1}{|G_{\delta,A}|},$$
	where $G_{\delta,A}=\{g\in G_\delta:g(A)=A\}$ and the second equality follows from the fact that $G_\delta$ is sharply transitive on $\Fp\backslash \{\delta\}$. Again by the sharp transitivity of $G_\delta$   there exists  exactly one affine transformation $g\in G_{\delta,A}$ such that $g(a_1)=a_2$ for any two $a_1,a_2\in A$, therefore, $|G_{\delta,A}|\leq |A|=d$ and 
	$|G_\delta^A|\geq (p-1)/d=\Omega(p).$    
	
	Next, let be $B\in G_\delta^A$, where $g(A)=B$, then the affine transformation $g\circ h$ satisfies $g\circ h(1)=\delta$ 
	and $g\circ h([r+1,r+d])=B$. Next, consider the punctured code of the full-length RS code, defined by the $d+1$ evaluation points $\delta$ and $B$. Since RS codes are invariant under linear transformation, the code can be viewed as a RS code with evaluation set $1$ and $[r+1,r+d]$ due to the affine transformation $g\circ h.$ Hence, repairing the failed node $\delta$ with helper nodes $B$ is equivalent to repairing the failed node $1$ with helper nodes $[r+1,r+d]$. By Theorem   		\ref{thm:n-k-d-allmost-msr} this can be done with bandwidth at most $d/(d-k+1)\log(p)+O_{k,d}(1)$ bits.  
	Note that we invoked Theorem \ref{thm:n-k-d-allmost-msr} with the punctured code $[2d, k]_p$ RS code, and thus the term $O_{n,k,d}(1)$ is in this case $O_{k,d}(1)$.
	Hence, the set $B\in G_\delta^A$ is indeed a valid set of helper nodes, and the result follows.  
\end{proof}

\subsection{Repairing by any $d$ helper nodes}
In the quest of constructing an RS code over a prime field that is asymptotically MSR, we give a construction of a folded RS code that follows from Construction \ref{cnst}. This new construction falls short in achieving this goal in two aspects compared with  Construction \ref{cnst}. First, the bandwidth is not asymptotically optimal, and second, the alphabet is not of  prime order, albeit  very close to it. However, it allows us to repair each failed node with any $d\geq k$ helper nodes and very small bandwidth.   
The  construction is derived from  \Cref{cnst} by a  simple folding operation, and it was observed by Amir Shpilka (who kindly allowed us to include it here). 

\begin{cor} \label{cor:folded-rs}
	Let $p$ be a large enough prime. Let $k,d$, and $n$ be positive integers such that $2k < d < n$, then there exists an $[n,k]_{p^2}$ folded-RS code such that any node can be repaired using \emph{any} $d$ helper nodes with bandwidth of at most $(2d/(d-2k+1)) \log(p)+O_{n}(1)$ bits.
\end{cor}
\begin{proof}
	Consider  the $[2n,2k]_p$ RS code provided in \Cref{cnst}, and recall that it is defined by the evaluation points $[1,n]\cup [r+1,r+n]$. Next, define the $[n,k]_{p^2}$ folded-RS as follows. For a polynomial $f\in \Fp[x]$ of degree less than $2k$ define a codeword with $i$th super-symbol to be $(f(i),f(i+r))$ for $i=1,\ldots,n$. 
	
	
	If  the $i$th symbol fails, then any set of $d$ 
	nodes can repair it. Indeed, fix a set of helper nodes $D\subseteq [n],|D|=d$ then the values $f(i)$ and  $f(r+i)$ can be repaired by the values $f(j),j\in D$ and 
	$f(j+r),j\in D$, respectively. The claim on the total bandwidth follows from the bandwidth guaranteed in Theorem \ref{thm:n-k-d-allmost-msr}, and the result follows. 
\end{proof}
Note that  by the cut-set bound, the lower bound on the repair bandwidth for codes with these parameters is at least  $(2d/(d - k + 1)) \log(p)$ bits. Hence, the guaranteed repair bandwidth given in Corollary \ref{cor:folded-rs} is not asymptotically optimal,   yet, for large values of $d$ compared with  $k$, it is very  close to it.

The folded RS construction presented above can be viewed as an array code with subpacketization $L=2$ over the field $\Fp$. We are aware of only one more code construction with such a small subpacketization level over its prime field $\Fp$ and with an efficient repair scheme. We state this result next,  rephrased  to fit this  context.

\begin{thm} \cite[Theorem 10]{guruswami2017repairing}
	\label{guru-wootters-res}
	Let $p$ be a prime. For any $n \leq 2(p - 1)$, there is an $[n,k]_{p^2}$ RS code and linear repair scheme such that any failed node can be repaired using the $n-1$ remaining nodes with bandwidth $$\left(\frac{3n}{2} - 2\right) \log (p).$$
\end{thm}
By setting $d = n-1$ in \Cref{cor:folded-rs} we obtain a  bandwidth of
\[
\frac{2(n-1)}{n-2k} \log(p) + O_n(1) \;,
\]
which is significantly smaller. However, \Cref{cor:folded-rs} requires that  $k < n/2$ while \Cref{guru-wootters-res} holds  for any $k\leq n-2$.

\section{An improved lower bound on the bandwidth}
\label{sec:new-cut-set-bound}
All of the constructions presented in this paper do not achieve the cut-set bound \eqref{eq:cut-set}, and the incurred  bandwidth is larger  than it   by an additive factor that depends on $k$ or $n$. Hence, we ask if this is indeed necessary, i.e., whether the cut-set bound is not tight for RS codes over prime fields, and if not, can it be improved?

We answer this question  in the affirmative by showing that the bandwidth is at least $d\log(p)/(d-k+1) +\Omega(d\log(k)/(d-k+1))$
which is an improvement over the cut-set bound by an additive factor of  $\Omega(d\log(k)/(d-k+1))$. 
Equivalently, by the terminology of leakage-resilient of  Shamir Secret Sharing (See Section \ref{SSS} for definitions),  this result provides an improved lower bound on the amount of information that needs to be leaked from the shares to the  adversary to recover the secret fully. 

Before stating and proving the main  result, we shall recall the following well-known theorem from additive combinatorics used in the proof.
\begin{thm} [Cauchy-Davenport inequality] \cite{cauchy1812recherches,davenport1935addition}
	If $p$ is a prime and $A,B\subseteq \mathbb{Z}_p$ then,
	$\left|A + B\right| \geq \min(|A| + |B| - 1, p)$.
\end{thm}


\begin{thm}
	\label{thm-improved-cut-set-bound}
	Consider a $k$-dimensional RS code over a prime field $\Fp$, and  assume  that there is a repair scheme for node $\alpha \in \Fp$  by the $d$ helper nodes 
	$\alpha_1, \ldots, \alpha_d \in \Fp$ where each helper node transmits the same amount of information. 
	Then, the bandwidth for each helper node is at least  
	\[
	\frac{\log(kp) - 1}{d- k +1} \;
	\]
	bits. 
\end{thm}

\begin{proof}
	For $i=1,\ldots,d$, let   $\mu_i: \Fp \rightarrow [s]$ be the function calculated by the helper node (symbol) $\alpha_i$, and let 
	$\mathcal{A}_i$ be the partition of $\Fp$ defined by $\mu_i$, i.e., $$\mathcal{A}_i=\{\mu_i^{-1}(m):m\in [s]\}.$$
	
	It is clear that since the $\mu_i$'s define a repair scheme, then for 
	any choice of sets $A_1 \in \mathcal{A}_1, \ldots, A_d \in \mathcal{A}_d$, all the polynomials that pass through $(\alpha_1, A_1), \ldots, (\alpha_d, A_d)$  attain the same value at  $\alpha$. Note that it might be possible that there  are no such polynomials.

	Fix for $i=1,\ldots, k $, 
	$A_i\in \mathcal{A}_i$ of size at least  $p/s$, and denote by $\mathcal{F}$  the set of polynomials of degree less than $k$, that pass through $(\alpha_1, A_1), \ldots, (\alpha_k, A_k)$. Since the code dimension is $k$, the size of $\mathcal{F}$  is exactly  $\prod_{i=1}^{k} \left|A_i\right|$. 
	Next, set  $U=\{f(\alpha):f\in \mathcal{F}\}$ to be  the set of  values attained by the polynomials in $\mathcal{F}$ at $\alpha$. We have the following claim on the size of $U$. 
	\begin{claim}
		$\left| U \right| \geq k p/s - k$.
	\end{claim}
	\begin{proof}

		Let $f = \sum_{i=0}^{k-1}f_i(x - \alpha)^i$ be a polynomial in $\mathcal{F}$,  and note that $f(\alpha)=f_0.$
		By abuse of notation  let also  $f = (f_0,\ldots, f_{k-1})$ to be the vector of coefficients of $f$.  
		Let  $V$ be the Vandermonde matrix defined by the $k$ distinct elements  $\alpha_1 - \alpha, \ldots, \alpha_k - \alpha$, i.e., 
		$V_{ij}=(\alpha_j-\alpha)^{i-1}$, where $i,j=1,\ldots, k$, and let 
		$A_1\times \ldots\times  A_k=\{(a_1,\ldots,a_k):a_i\in A_i)\}$. Thus,    $ f \in  \{wV^{-1}:w\in A_1\times  \ldots \times A_{k}\}$ for any $f\in \mathcal{F}$, and in particular 
		\begin{equation} \label{eq:sum-set-cut-set}
		U = \sum_{i = 1}^{k}  A_{i}(V^{-1})_{i1}\;.
		\end{equation}
		We claim that the first column of $V^{-1}$ has nonzero entries. Indeed, the $i$th row of the inverse Vandermonde matrix corresponds to the coefficients of a polynomial $p(x)$ that vanishes at $\alpha_j - \alpha$ for $j \in [k] \setminus \{i\}$ and $p(\alpha_i - \alpha) = 1$. Hence, the first entry of the row equals $p(0)$ which is clearly nonzero since $p(x)$ is a nonzero polynomial of degree less than $k$ with   $k-1$ nonzero roots. 
		
		Now, by applying $k$ times the Cauchy-Davenport inequality we get that  $\left| U \right| \geq \sum_{j = 1}^{k} \left| A_{j} \right| - k \geq k p/s - k$.
	\end{proof}
	
	We conclude that there are at least $k p/s - k$ polynomials $g_i$ in $\mathcal{F}, i=1,\ldots, k p/s - k$ that attain distinct values at $\alpha.$ 
	Define a mapping $\mathcal{F}\rightarrow  \mathcal{A}_{k+1}\times \ldots\times \mathcal{A}_{d}$ by $f\mapsto (A_{k+1}',\ldots,A_{d}')$, where $A_{i}'$ is the set in the partition $\mathcal{A}_{i}$ that contains $f(\alpha_i).$ Clearly, this map has to be injective on the polynomials $g_i$, otherwise the repair scheme would not be able to repair the failed node. 
	Hence \[
	s^{d - k} \geq k \cdot \frac{p}{s} - k  \geq \frac{kp}{2s}. 
	\] 
	By rearranging, we get that 
	\[
	\log(s) \geq \frac{\log(kp) - 1}{d-k+1} \;.
	\]
\end{proof}

The following theorem is an inverse theorem to the Cauchy-Davenport inequality, which characterizes the sets that attain the bound with equality. 
\begin{thm} \cite[Vosper's theorem]{vosper1956critical}
	\label{vosper-thm}
	Let $A,B\subseteq \mathbb{Z}_p$ where $|A|,|B| \geq 2$ and $|A+B| \leq p-2$. Then $|A+B| = |A| + |B| - 1$ if and only if $A$ and $B$ are arithmetic progressions with the same step.
\end{thm}
By examining the improved bound's proof, it is clear that   a large set $U$ in the proof of Theorem \ref{thm-improved-cut-set-bound}  would  imply a strong lower bound on the bandwidth. Considering the extreme case where we replace the size of $U$ with the trivial lower bound of $p/s$, we recover   the cut-set bound. Therefore,  the  improvement follows from sumsets' expansion phenomenon over prime fields, exemplified by the     Cauchy-Davenport inequality. 
On the other hand, it is known that the cut-set bound is tight over large field extensions, which implies that the trivial lower bound of $p/s$ is, in fact, tight. And indeed, field extensions do contain nontrivial subsets whose sumset do not exhibit any expansion, i.e.,  subsets $A, B$ such that $|A+B| = |A|$.  
Hence, we arrive at the following conclusion. If one wants to construct efficient repair schemes for RS codes over prime fields, then the size of the set $U$ should have little expansion as possible. Identifying the structure of such sets (a.k.a. sets with small doubling constant) is a well-studied problem in additive combinatorics,  known as the   \emph{inverse sum set problem}. Over prime fields, Vosper's theorem (Theorem \ref{vosper-thm}) shows that the only sets that attain the  Cauchy-Davenport lower bound with equality are arithmetic progressions 
with the same step. This fact lies at the core of all the constructions given in this paper, as arithmetic progressions define all the functions computed by the helper nodes with the same step. For example, consider the $[n,2]$ RS code given in   \Cref{sec:n-2-existence}, and assume  that we would like to repair the  $n$th symbol. Then, a simple computation shows that by the partitions defined for the repair scheme, the set $U$ takes the following form  
\[
U = \frac{(\alpha_1 - \alpha_n) (\alpha_2 - \alpha_n)}{\alpha_1 - \alpha_2} (A - B) 
\]
where $A$ and $B$ are both   arithmetic progressions of size $t$ and step $1$, and therefore $|U| = |A| + |B| - 1$, which is as small as possible. 

As a final remark, one might ask  whether a random repair scheme is expected to have low bandwidth. By a random scheme, we mean that the function computed by  each helper node is randomly picked among all  functions with a fixed range size, which  is equivalent to picking a random partition of the field to a fixed number of sets. One can  verify that in such a case, with high probability, the size of $U$ is as large as possible since the sumset of two random subsets is large with high probability. Thus, the event of picking an efficient repair scheme is improbable, and one needs to carefully construct the repair scheme to obtain low bandwidth.


\section{Connection to leakage-resilient Shamir's Secret Sharing}
\label{SSS}
Shamir's Secret Sharing (SSS), which was first introduced in \cite{shamir1979share} is a fundamental cryptographic primitive that provides a secure method to   distribute a secret among   different parties, and it shares a lot of similarities with RS codes, as they both  rely on the same algebraic object, evaluations of bounded degree polynomials \cite{mceliece1981sharing}. 
The formal definition of SSS is  as follows. Given a secret $s\in \mathbb{F}$ to be distributed, a reconstruction threshold $k>0$ and $n$ parties (identified as $n$ distinct elements $\alpha_i\in \mathbb{F}$),  a dealer randomly selects a polynomial $f\in \mathbb{F}[x]$ of degree less than $k$ such that $f(0)=s$, and sends to party $i$ the share $f(\alpha_i).$  It is easy to verify that any $k$ parties can reconstruct the polynomial $f(x)$ and therefore recover the secret $f(0)=s$, although the shares of any $k-1$ parties reveal no information about the secret.  


Inspired by the work of Guruswami and Wootters \cite{guruswami2017repairing} on  repairing  RS codes,  
recently Benhamouda, Degwekar,   Ishai,  and  Rabin
\cite{benhamouda2018local} considered the question  of local leakage-resilience of secret sharing schemes and, in particular, SSS. 
In this setting, an adversary who wishes to learn the secret can  apply an arbitrary leakage function of small image size to each party's share. The question is how much information on the secret is leaked to him by doing so. Hence, efficient repair schemes of RS codes can be viewed in the context of secret sharing as a very poor local leakage resistance of SSS since the adversary can fully determine the secret.   
For example,  in \cite{guruswami2017repairing} it was shown that there are full-length RS codes over field extensions such that even one bit of information from each  non-failed node suffices to repair the failed node. Equivalently, in the terminology of  local leakage resilience of SSS,  the adversary can learn the secret by applying on the shares leakage functions  whose output is only a single bit. In this paper, we showed that RS codes over prime fields also exhibit the same kind of phenomenon, which translates to the fact that SSS is not local leakage over prime fields for some parameters regime.   
On the other hand, \cite{benhamouda2018local} showed that SSS is leakage resilient for some parameter regimes over prime fields. Hence, at first glance, it might seem that the results presented in this paper and the results of \cite{benhamouda2018local} contradicting one another. Of course, this is not  true,  and in this section, we would like to give a clear picture of how these two sets of results align with each other. To this end, we recall some definitions from \cite{benhamouda2018local}.

Assume that the adversary has full access to shares $s_i,i \in \Theta$ for some $\Theta \subseteq [n]$, and for the remaining shares $s_i$ for $i\notin \Theta$ he has an access to $\tau_i(s_i)$, where $\tau_i$ is an arbitrary  leakage function that leaks exactly $m$ bits. In other words, the range of the function $\tau_i$ is of  size $2^m$.
Formally, the adversary has access and therefore learns
\[
\Leak = \left(\left(s_i\right)_{i\in \Theta}, \left(\tau_i(s_i) \right)_{i\in [n]\setminus \Theta} \right).
\]
\begin{defi}
	A secret sharing scheme is said to be $(\theta, m, \varepsilon)$-local leakage resilient if for any subset $\Theta \subset [n]$, where $|\Theta|=\theta$, for every choice of leakage functions $\tau_i,i\notin \Theta$ that leak $m$ bits,  and for every pair of secrets $s, s'$, it holds that
	\[
	\textup{SD}\left( \{\Leak({\bf s}) : {\bf s} \leftarrow \textup{Share}(s) \} , \{\Leak({\bf s'}) : {\bf s'} \leftarrow \textup{Share}(s') \}\right) \leq \varepsilon
	\]
	where $\textup{SD}$ is the statistical distance and $\textup{Share}(s)$ is a random choice of shares that correspond to the secret $s$.
\end{defi}

Next, we state the results of \cite{benhamouda2018local} that are relevant to us.  
The first result applies to when a constant number of bits $m$ is leaked from each share.
\begin{thm} \label{thm:secret-sharing-leakage-cons-rate} \cite[Corollary 4.12]{benhamouda2018local}
	If $m = O(1)$, $\theta = O(1)$, and $n$ goes to infinity, there exists $\alpha < 1$, such that the Shamir's secret sharing scheme with $n$ players and threshold $k\geq \alpha n$ is $(\theta, m, \varepsilon)$-local leakage  resilient where $\varepsilon = 2^{-\Omega(n)}$
\end{thm}

We are interested in the case where $\theta = 0$, i.e., the adversary does not have full access to any share. 
We compare this theorem with \Cref{thm:full-length-rs}. Recall that in \Cref{thm:full-length-rs}, we show that any node in the $[p,k]_p$ RS code (the full length RS code) admits an asymptotically optimal repair using $d>k$ helper nodes.

In the context of secret sharing, \Cref{thm:full-length-rs} basically says that there is a strategy for an adversary to learn the secret completely by applying some specific leakage functions to some specific shares. More formally, 
we show that there is an adversary that can pick a set of $d$ shares $\alpha_1, \ldots, \alpha_{d}$ and $d$ leakage functions $\tau_{\alpha_1}, \ldots, \tau_{\alpha_{d}}$ that output $(1/(d-k+1)) \log (p) + O_{k,d}(1)$ bits such that when applied to $f(\alpha_{1}), \ldots, f(\alpha_{{d}})$, the adversary completely learns the secret. In fact, as was prove in \Cref{thm:full-length-rs}, the adversary has $\Omega(p)$ distinct sets of $d$ shares from which he can completely learn the secret.

Note that it does not contradict \Cref{thm:secret-sharing-leakage-cons-rate} since we are in a totally different parameter regime. Indeed, the first observation is that in \Cref{thm:secret-sharing-leakage-cons-rate}, the leakage functions output a constant number of bits, while in our settings, we output a constant fraction of bits.
Secondly, in their framework, the code's rate is constant, and thus $k$ grows with $n$ while in our work, it is a constant compared to $n = p$.

The second result gives another parameter regime for which Shamir's secret sharing scheme is leakage resilient.
\begin{thm} \label{thm:secret-sharing-leakage-high-rate} \cite[Corollary 4.13]{benhamouda2018local}
	Let $\theta = O(1)$. For sufficiently large $n$, where $n < p \leq 2n$ and $m = \floor{(\log (p))/4}$, the Shamir's secret sharing scheme with $k > n-n^{1/4}$ is $(\theta, m, \varepsilon)$-local leakage resilient where $\varepsilon = 2^{-\Omega(n)}$.
\end{thm}
This parameter regime corresponds to a very high rate code. They show that in this case, the adversary can read even a quarter of the bits of every share and yet learn nothing about the secret. 


We note that by puncturing the code in \Cref{thm:full-length-rs}, one can simply prove that there exists an $[k+4,k]_p$ RS code such that a single node (the last node) can be repaired by downloading $(1/4)\log(p) + O_k(1)$ bits from each of the remaining $k+3$ nodes.
This of course does not contradict \Cref{thm:secret-sharing-leakage-high-rate} since in our results, $p$ is very large compared to $n = k + 3$ while in \Cref{thm:secret-sharing-leakage-high-rate}, it is required that  $n<p\leq 2n$.

\section{Concluding remarks and open problems}
\label{sec:concluding remarks}
The study presented in this paper was inspired by the interesting open question raised in \cite{guruswami2017repairing} regarding whether nonlinear repair schemes exist, and if so,  can they  outperform linear schemes. 
Since for codes over prime fields, any linear repair scheme is the trivial scheme, any efficient repair scheme, i.e., a repair scheme that outperforms the trivial one, must be nonlinear. Hence, our primary  focus was on constructing  repair schemes of  RS codes  over prime fields.

We were able to exhibit the first nonlinear repair scheme of RS codes over prime fields, which is also the only known example of a nonlinear repair scheme of any code. As a byproduct, we showed that nonlinear ones can outperform linear repair schemes. 
Furthermore, some of the repair schemes are  asymptotically optimal, as the alphabet size tends to infinity.  
Lastly, we also improved the cut-set bound for RS codes over prime fields  and discussed connections to leakage-resilient Shamir's Secret Sharing over prime fields. 

We end this discussion by mentioning several open questions that, in our opinion, could further improve the study of nonlinear repair schemes.

\begin{enumerate}
	\item Is it possible to apply our approach of using arithmetic progressions to obtain RS codes over prime fields that are asymptotically MSR codes for any positive integers $k < d < n$? Note that it is unknown if such codes exist, although they likely do. Moreover, is it possible to  generalize the  approach to other codes over prime fields?
	
	\item The work of \cite{benhamouda2018local} showed that for some parameters, an adversary learns almost nothing on the secret in SSS, whereas we showed in this paper the other extreme case. Namely, for some other parameters, the adversary learns the entire secret.  Hence, it is  interesting to fill in the gaps and improve our understanding of SSS performance under the remaining parameters regime. In particular, better  understand the dependence between the field size $p$ and the number  of bits $m$ leaked to  the adversary to learn something or the whole secret. For example, what can be said for $m = O(\log(p))$, $k = O(n)$, and $p > 2n$. Notice that any new result would automatically have implications on the other
	model of repairing RS codes.
	
	\item It is known  \cite{tamo2017optimal,alrabiah2019exponential}   that linear MSR code with a linear repair exists only over alphabet size, which is at least doubly exponential in the code dimension. Can this result be generalized to nonlinear repair schemes? In particular, does the field size have to be large for efficient repair schemes to exist over prime fields? 
\end{enumerate}

\bibliographystyle{alpha}
\bibliography{repair_RSp}
\section{Appendix}
\subsection{Proof of 
	\label{proof-of-the-lemma}
	\Cref{lem:lcm-lemma}}
\begin{proof}
	We prove it by induction on $s$. If $s=2$ it holds that 
	\[
	\lcm(a_1, a_2) = \frac{a_1 \cdot a_2}{\gcd (a_1, a_2)} \;.
	\] 
	Assume that the claim holds for $s-1$. It holds that
	\begin{align}
	\lcm(a_1, \ldots, a_s) &= \lcm(a_1, \lcm(a_2, \ldots, a_s)) \nonumber \\
	& = \frac{a_1 \cdot \lcm(a_2, \ldots, a_s)}{\gcd(\ a_1, \lcm(a_2, \ldots, a_s))} \nonumber \\
	& \geq \frac{a_1 \cdots a_s}{\gcd( a_1, a_2 \cdots a_s) \cdot \prod_{2\leq i < j\leq s} \gcd( a_i, a_j)} \label{eq:hypo} \\
	& \geq \frac{a_1\cdots a_s}{\prod_{i=2}^{s} \gcd( a_1, a_i) \prod_{2\leq i < j\leq s} \gcd( a_i, a_j)} \label{eq:mult-gcd} \\
	& = \frac{ a_1 \cdots a_s}{\prod_{1\leq i < j\leq s} \gcd ( a_i, a_j)} \nonumber \;.
	\end{align}
	Inequality \eqref{eq:hypo} follows from the induction hypothesis and from the fact that $\lcm( a_2, \ldots, a_s) \mid a_2 \cdots a_s$ which implies that 
	\[
	\gcd ( a_1, \lcm( a_2, \ldots, a_s)) \leq \gcd( a_1, a_2 \cdots a_s) \;.
	\]
	Inequality \eqref{eq:mult-gcd} follow from \Cref{clm:gcd-bound} which is stated and proved below.
\end{proof}

\begin{claim} \label{clm:gcd-bound}
	Let $b,a_1,\ldots, a_s$ be integers. It holds that 
	\[
	\gcd(b, a_1\cdots a_s) \leq \gcd(b,a_1) \cdot \gcd(b, a_2) \cdots \gcd(b, a_s)
	\]
\end{claim}
\begin{proof}
	It is easy to check that $\gcd(b, a_1a_2)$ divides the product $\gcd(b, a_1) \cdot \gcd(b,a_2)$. 
	Thus, by induction on $s$, we get that $\gcd(b,a_1\cdots a_s) | \gcd(b, a_1 \cdots a_{s-1}) \gcd(b, a_s) | \gcd(b,a_1) \cdots \gcd(b, a_s)$. 
\end{proof}

\subsection{Repairing $\alpha_i$ for $i\in \{1,2,3\}$} \label{sec:repairing-rest}
Recall the four evaluation points
\[
\alpha_1 = 0, \alpha_2 = -1, \alpha_3 = \frac{p-1}{2}, \alpha_4 = -(2t + 1)\;,
\]
where we assume that $t = \floor{\sqrt{p}/5} $.
In the following, we show that Proposition \ref{con:repair-condition} holds also for all the remaining symbols.

Assume that we wish to repair the $i$th node for some $i\in \{1,2,3\}$ and define $\gamma_j = \alpha_j - \alpha_i$ for every $j\in [4]\setminus \{i\}$. By Proposition \ref{con:repair-condition}, the repair of $\alpha_i$ succeeds if for any polynomial $f$ of degree at most one, that satisfies $f(\alpha_j)\in \gamma_j \cdot [-t,t]$ for $j\in [4]\setminus \{ i \}$, it holds that $f(\alpha_i) = 0$. Let $f$ be such a polynomial, i.e., $\deg(f)\leq 1$ and $f(\alpha_j)\in \gamma_j\cdot [-t,t]$ for $j=[4] \setminus \{i \}$, write $f(\alpha_4)=m\cdot \gamma_4$ for some $m\in [-t,t]$, and consider the polynomial $\hat{f}:=m(x-\alpha_i)$. We will show that $f(\alpha_i) = \hat{f}(\alpha_i) = 0$, as needed.  

Define $g := f - \hat{f}$ and denote $\{k, \ell \} = [3] \setminus \{ i \}$. It holds that 
\begin{align*}
g(\alpha_k) &\in \gamma_k \cdot [-2t,2t] \\
g(\alpha_{\ell}) &\in \gamma_{\ell} \cdot [-2t,2t] \\
g(\alpha_4) &= 0 \\
g(\alpha_i) &= f(\alpha_i) \;. 
\end{align*}
Thus, $g(x) = s(x - \alpha_4)$ and calculating, we get
\begin{align*}
s &\in \frac{\alpha_k - \alpha_i}{\alpha_k - \alpha_4} [-2t,2t] \\
s & \in \frac{\alpha_{\ell} - \alpha_i}{\alpha_{\ell} - \alpha_4} [-2t,2t] \;.
\end{align*}

We will show that if a certain condition is satisfied, then there are no $a, b\in [-2t,2t]\setminus\{ 0 \}$ such that
\begin{equation} \label{eq:con-2-4-constraint-sets}
(\alpha_{\ell} - \alpha_4)(\alpha_{\ell} - \alpha_i)^{-1} \cdot a = (\alpha_k - \alpha_4) (\alpha_k - \alpha_i)^{-1} \cdot b
\end{equation}
which implies that $s = 0$, as needed.

Assume that we take a realization of $\mathbb{F}_p$ as all the integers whose absolute value are less than $p/2$, i.e., $F_p=\{0,\pm 1,\pm 2,\cdots, \pm \frac{p-1}{2}\}$. Assume that the absolute value of the products $(\alpha_k - \alpha_4) (\alpha_k - \alpha_i)^{-1} \cdot 2t$ and $(\alpha_{\ell} - \alpha_4)(\alpha_{\ell} - \alpha_i)^{-1} \cdot 2t$ are less than $p/2$. Hence, in such a case the calculation in equation \eqref{eq:con-2-4-constraint-sets} holds over $\mathbb{Z}$.

Lastly, if the $\lcm$ of $|(\alpha_k - \alpha_4) (\alpha_k - \alpha_i)^{-1}|$ and $|(\alpha_{\ell} - \alpha_4)(\alpha_{\ell} - \alpha_i)^{-1}|$ is greater than $2t\cdot \min(|(\alpha_k - \alpha_4) (\alpha_k - \alpha_i)^{-1}|, |(\alpha_{\ell} - \alpha_4)(\alpha_{\ell} - \alpha_i)^{-1}|) $ (again we view them as integers), then it is easy to verify that there are no $a$ and $b$ in $[-2t,2t]\setminus \{ 0 \}$ such that equation \eqref{eq:con-2-4-constraint-sets} holds. Therefore, it must be that $a=b=0$ which implies that $s=0$, and we are done.  

Now, note that there is symmetry between $k$ and $\ell$, thus we check only the following three options:
\begin{itemize}
	\item  $k = 1, \ell = 2, i = 3$. In this case, we get that 
	$|(\alpha_k - \alpha_4) (\alpha_k - \alpha_i)^{-1}| = 4t + 2$,  $|(\alpha_{\ell} - \alpha_4)(\alpha_{\ell} - \alpha_i)^{-1}| = 4t$, and $\lcm(4t +2, 4t) = 4t  \cdot (2t+1)$.
	
	\item  $i = 1, \ell = 2, k = 3$. In this case, we get that
	$|(\alpha_k - \alpha_4) (\alpha_k - \alpha_i)^{-1}| = 4t - 1$,  $|(\alpha_{\ell} - \alpha_4)(\alpha_{\ell} - \alpha_i)^{-1}| = 2t$, and $\lcm(4t -1, 2t) = 2t\cdot (4t+1)$. 
	
	\item  $k = 1, i = 2, \ell = 3$. In this case, we get that
	$|(\alpha_k - \alpha_4) (\alpha_k - \alpha_i)^{-1}| = 2t + 1$,  $|(\alpha_{\ell} - \alpha_4)(\alpha_{\ell} - \alpha_i)^{-1}| = 4t + 1$, and $\lcm(4t + 1, 2t+1) = (4t + 1) \cdot (2t+1)$. 
	
\end{itemize}
In all cases, one can verify that $|(\alpha_k - \alpha_4) (\alpha_k - \alpha_i)^{-1} \cdot 2t|,  |(\alpha_{\ell} - \alpha_4)(\alpha_{\ell} - \alpha_i)^{-1} \cdot 2t| < p/2$. Furthermore, in all cases, it holds that 
the $\lcm$ of $|(\alpha_k - \alpha_4) (\alpha_k - \alpha_i)^{-1}|$ and $|(\alpha_{\ell} - \alpha_4)(\alpha_{\ell} - \alpha_i)^{-1}|$ is strictly greater than $2t\cdot \min(|(\alpha_k - \alpha_4) (\alpha_k - \alpha_i)^{-1}|, |(\alpha_{\ell} - \alpha_4)(\alpha_{\ell} - \alpha_i)^{-1}|) $. We conclude that $\alpha_i$ can be repaired with the desired bandwidth.
\end{document}